\documentclass[submission,copyright,creativecommons]{eptcs}

\providecommand{\event}{TARK 2025} 

\usepackage{iftex}

\ifpdf
  \usepackage{underscore}         
  \usepackage[T1]{fontenc}        
\else
  \usepackage{breakurl}           
\fi

\usepackage{graphicx} 
\usepackage[dvipsnames]{xcolor}
\usepackage{amsmath,amssymb,amsfonts}
\usepackage{amsthm}
\usepackage{hyperref}
\usepackage{placeins}
\usepackage{graphicx}
\usepackage{tikz}
\usetikzlibrary{shapes.geometric,arrows}
\usepackage{makeidx}

\newtheorem{theorem}{Theorem}
\newtheorem{Corollary}{Corollary}

\newtheorem{definition}{Definition}
\newtheorem{proposition}{Proposition}

\renewcommand{\phi}{\varphi}

\allowdisplaybreaks[4]

\title{Are Large Random Graphs Always Safe to Hide?}
\author{Sourav Chakraborty
\institute{Indian Statistical Institute}
\email{sourav@isical.ac.in}
\and
Sujata Ghosh
\institute{Indian Statistical Institute}
\email{sujata@isichennai.res.in}
\and
Smiha Samanta
\institute{Indian Statistical Institute}
\email{smi1995ha@gmail.com}
}

\def\titlerunning{Are Large Random Graphs Always Safe to Hide?}
\def\authorrunning{Sourav Chakraborty, Sujata Ghosh \& Smiha Samanta}

\begin{document}
\maketitle


\begin{abstract}
    We discuss winning possibilities of players in various variants of cops and robber game played on large random graphs, a testbed for various kinds of network queries, search problems in particular. We explore the use of logic frameworks to investigate such results; in particular, we show that whenever a winning condition for either player can be expressed as a certain kind of formula in first-order logic, that player almost always wins. In the process, we obtain more insight into the logic-game connection from the zero-one law perspective.
\end{abstract}


\section{Introduction}

We begin with describing the cops and robber game \cite{cop-robber-book} in detail, which is the focus of the current study. In the cops and robber game, we have a finite ($\geq 1$) number, say $k$, of cops looking for a robber on a game arena, provided by a graph structure. The players occupy vertices of the graph, and they move along the edges to adjacent vertices. We assume that the underlying graph is reflexive and, accordingly, staying at her vertex can also be considered as a move for the player. This is a turn-based game with the cops and the robber occupying certain vertices and moving alternately, and in their turn, the cops move simultaneously, where some may stay put on their vertices. The game is played in a countable sequence of rounds. 
The cops win if after some finite number of rounds, at least one of the cops can occupy the same vertex as the robber. The robber wins if he can avoid the cops in an infinite run. For this work, as is the traditional case, players are assumed to have perfect information in these games.

The cops and robber games that are played on graphs are proposed in \cite{quilliot1978,NOWAKOWSKI1983,AIGNER1984}, and since then they have been extensively studied in various forms from distinct viewpoints. From the algorithmic and combinatorial perspectives, the central point of the study is the \emph{cop number} of the game, that is, the minimum number of cops required to win a game. For a detailed summary of these studies, see \cite{cop-robber-book}. In the study of robotic systems, search missions in a pursuit-evasion environment, in particular adversarial search problems, have been a primary object for investigations \cite{p-e1,p-e2}. In recent years, logicians have studied the hide and seek game, which shares a similar pursuit-evasion structure, with a focus on modeling the dynamic interaction between players \cite{graphgame,LHS-journal}. 

These games, played on graphs, occur in many natural contexts where the underlying graphs are generated by some stochastic processes; for example, see \cite{stojakovi2006}. In \cite{sabotage-random}, sabotage games (introduced in \cite{vanbenthem2005}), played between a traveler and a demon, are investigated to check whether they are biased towards the traveler (who travels along the edges of the graph to reach her goal vertex) or the saboteur (who removes edges from the graph to prevent the traveler from reaching her goal), and it is found that the game is indeed biased towards the traveler in large random graphs. In this work, we investigate the cops and robber game to check whether the game conditions favor the cops or the robber when the size of the underlying graph structure (number of vertices) becomes larger and larger. Since we do not have any control over the occurrence of edges in the graphs, we consider random graphs \cite{erdos1959random,gilbert} with constant and varying edge probabilities. Thus, we deal with cops and robber games played on random graphs. 

We should note here that from the algorithmic perspectives (checking the cop number of the game), these games have already been studied with respect to random graphs (e.g., see \cite{bollobas}). Our goal here is to highlight how the notion of winning in these games is skewed towards one of the players when we consider these graphs as the corresponding game arenas. These results can be used for further developments with respect to the design and expansion of relevant networks. The general interest in studies featuring large networks has increased manifold, and the present study may help us  deal with pursuit-evasion problems in randomly generated networks. Our intuition tells us that the larger the graph size is, the easier it should be for the robber to hide inside. We will see below whether the formal results agree with our intuitions. In addition to the traditional cops and robber game, we also study various variants of the game played on random graphs to explore how the bias that we discussed earlier, varies in terms of these variants.

 The proofs that we have here are simple applications of various known results relating random graphs and first-order logic. Our main goal is to showcase the elegant connection between first-order logic, probability theory and graph games. We believe that this connection is worth-pursuing over various kinds of random models and intend to provide a glimpse of this rich field of study. Our proof techniques rely on the same for showing zero-one laws in first-order logic over graphs \cite{Rfagin1976}, and we first give a brief primer on those results.


\section{On zero-one laws in first-order logic}

The main idea behind the zero-one law, satisfied by some logic, is the feature that various properties that are expressible in the logic under consideration can be either \emph{almost surely true} or \emph{almost surely false}. In the following, we show that the first-order theory of graphs, whose vocabulary consists only of a binary relation symbol, satisfies the zero-one law. Let us first rigorously define what we mean by `almost surely', following \cite{fmt}.

Let $\mathit{GR_n}$ denote the set of all graphs with $n$ vertices given by $\{0, 1, \ldots, n-1\}$. Then, the number of possible graphs on these $n$ vertices is given by: $\vert\mathit{GR_n}\vert = 2^{\binom{n}{2}}.$ 

\begin{definition}
Let $\varphi$ denote some first-order formula expressing some property of the graphs, say $\mathsf{P}$. We define:
$$\mu_n(\varphi) = \frac{\vert\{G\in\mathit{GR_n} : G\;\textrm{satisfies}\;\varphi\}\vert}{\vert \mathit{GR_n} \vert}.$$
Thus $\mu_n(\varphi)$ is the probability that a randomly chosen graph on the set of nodes $\{0, 1, \ldots, n-1\}$ satisfies the property expressed by $\varphi$. 
We now deine the asymptotic probability of $\varphi$ as: $$\mu(\varphi) : =\lim\limits_{n\to\infty} \mu_n(\varphi).$$ 
\end{definition}

In the same way, without going into the underlying logic per se, we can define $\mu_n(\mathsf{P})$ and $\mu(\mathsf{P}).$ Let us now establish the following well-known result from first-order logic. We show this for the first-order language of graphs, $\mathcal{L}_{G}$. To begin with, we have the following definition of \emph{almost surely} true or false formulas.

\begin{definition}
A formula $\varphi \in \mathcal{L}_{G}$ is said to be \emph{almost surely true} if $\mu(\varphi) = 1$, and is said to be \emph{almost surely false} if $\mu(\varphi) = 0$.
\end{definition}

Accordingly, we have the zero-one law for the first-order logic of graphs. This is a well-established result; however, for a better understanding of the later results, we provide a proof idea below.

\begin{theorem}[\cite{Rfagin1976}]\label{fol}
    Any sentence in first-order logic with only a binary relation symbol in the vocabulary is either almost surely true or almost surely false.
\end{theorem}
\begin{proof}
The main steps of the proof involve the following: (i) defining a set of axioms termed as \emph{extension axioms} regarding how graphs with $n$ vertices can be extended to graphs with $n + 1$ vertices, and showing them to be surely true; (ii) constructing a countable model, $\mathsf{RG}$, say, of these axioms, unique upto isomorphism, and showing that for any sentence $\varphi$ in the given language, either $\mathsf{RG} \vDash \varphi$  or, $\mathsf{RG} \vDash \neg \varphi$; and finally, (iii) showing that for any sentence $\varphi$ with $\mathsf{RG} \vDash \varphi$, $\varphi$ is almost surely true. Thus we have that any sentence $\varphi$ is either almost surely true or almost surely false. 

    \textbf{Extension axioms:}
    The extension axioms simply state the fact that for any subgraph of size $n$, one can get a subgraph of size $n + 1$ and one can make the new vertex adjacent or non-adjacent to the pre-existing vertices in whatever way one wants. In case the new vertex is adjacent to $m$ pre-existing vertices, it can be expressed as follows:
    
   $$EA_{m,n} : \forall x_1 ,x_2 ...x_n (\bigwedge \limits_{1 \leq i < j \leq n} \neg (x_i= x_j) \rightarrow\exists z (\bigwedge \limits_{1 \leq i \leq n} \neg (z= x_i) \wedge \bigwedge \limits_{1\leq i \leq m} E(z,x_i) \wedge \bigwedge \limits_{m < i \leq n} \neg E(z,x_i)) )$$

   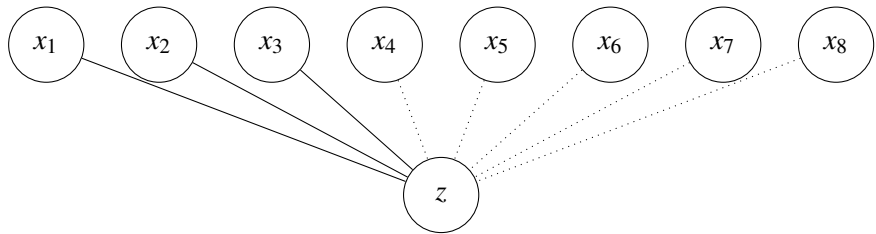
\begin{figure}[ht]
       \centering
       \begin{tikzpicture}[every node/.style={circle, draw, minimum size=1cm}, node distance=1.5cm]
\foreach \i in {1,...,8} {
    \node (x\i) at (\i*1.5, 0) {$x_{\i}$};
}

\node (z) at ({(8+1)*1.5/2}, -2) {$z$};

\foreach \i in {1,...,3} {
 
    \draw[-] (x\i) -- (z);
}
 \foreach \i in {4,...,8} {
   \draw[dotted] (x\i) -- (z);
   }
\end{tikzpicture}
\caption{$EA_{3,8}$ : for all possible choice of vertices $x_1, x_2, \ldots x_8$ we can always find a vertex $z$ that is adjacent to only 3 of those 8 vertices.}

       \label{fig:enter-label}
   \end{figure}
  
   \textbf{Extension axioms are almost surely true:} We consider $\mu_k (EA_{m,n})$, for $k > n$, and concentrate on $\mu_k (\neg EA_{m,n}).$ Our goal is to show that $\lim \limits_{k \to \infty}{\mu_k (\neg EA_{m,n})} = 0$, which would then imply $\lim \limits_{k \to \infty}{\mu_k (EA_{m,n})} = 1$. For $EA_{m,n}$ to be false in a $k$-size graph structure, there must be an $m$-size set of vertices, say $A$ and an $n - m$-size set of vertices, say $B$, such that no vertex in the complement of $A \cup B$ is adjacent to the vertices in $A$ and not adjacent to the vertices in $B$. For vertices in $A \cup B$, the possible  choices together with the choices for their adjacency are given by ${k\choose m} \times {{k-m} \choose {n-m}} \times 2^{n \choose 2}$. For the vertices outside $A\cup B $, the choice of adjacency is given by $2^{{k-n}\choose 2}$. Finally, for any vertex outside A$\cup$B, it can be adjacent to one in A$\cup$B in any way except for the desired one, and hence adjacency can be chosen in $2^n -1$ ways. Thus we have:

   $$\mu_k (\neg EA_{m,n}) \leq \frac{{k \choose m} \times {{k-m} \choose {n-m}} \times 2^{n \choose 2} \times 2^{{k-n}\choose 2} \times (2^n -1)^{k-n} }{2^{k \choose 2}} \leq \frac{k^n \times (2^n -1)^{k-n} }{2^{n \times(k-n)}} = k^n \times (1- \frac{1}{2^n})^{k-n}$$

   \noindent Thus, $\lim \limits_{k \to \infty}{\mu_k (\neg EA_{m,n})}=0 $, and hence $\lim \limits_{k \to \infty}{\mu_k (EA_{m,n})}=1 $, that is, 
  $\mu(EA_{m,n})= 1$ that is, $EA_{m,n}$ is almost surely true for all m,n $\in \mathbb{N}.$ 

  \textbf{Extension axioms have a unique countable model, say $\mathsf{RG}$:}
Let us consider $\mathsf{RG}$ to be a countable random graph, which has the set of natural numbers as its vertices and for any pair $(i,j)$ of vertices with $i < j$, we assume them to be adjacent with probability 1/2. We claim that $EA_{m,n}$ holds in $\mathsf{RG}$. If $EA_{m,n}$ does not hold at $\mathsf{RG}$, then there are $n$ vertices such that any other vertex can be related to them in all but one way (given by $EA_{m,n}$). The probability that one vertex cannot be related in the desired way is $(1-\frac{1}{2^n}).$ Thus, the probability that no vertex is related in the desired way is given by $\lim \limits_{k \to \infty}{(1-\frac{1}{2^n})^k} = 0$. Hence, $EA_{m,n}$ holds in $\mathsf{RG}$, and this completes the proof of our claim. That such a countable random graph is unique up to isomorphism can be shown by an inductive argument.

 \textbf{For any sentence $\varphi$ either $\mathsf{RG} \vDash \varphi$  or, $\mathsf{RG} \vDash \neg \varphi$ :}
Let us suppose that there is a sentence $\varphi$ such that $\mathsf{RG} \nvDash \varphi$  and $\mathsf{RG} \nvDash \neg \varphi$.
Then we can consider a model for all the extension axioms with $\varphi$, say $\mathsf{G^\prime}$, and another one
for all the extension axioms with $\neg\varphi$, say $\mathsf{G^{\prime\prime}}$. But then, 
$\mathsf{G^\prime}\cong \mathsf{G^{\prime\prime}} \cong \mathsf{RG}$, as both $\mathsf{G^\prime}$ and $\mathsf{G^{\prime\prime}}$ satisfy all the extension axioms.
This implies that $\mathsf{RG} \vDash \varphi$  and $\mathsf{RG} \vDash \neg \varphi$, a contradiction. 

\textbf{For any sentence $\varphi$ with $\mathsf{RG} \vDash \varphi$, $\varphi$ is almost surely true :}
As $\mathsf{RG} \vDash \varphi$, by compactness, there is a finite subset of extension axioms, $\mathit{Fin}$, say,
such that $\mathit{Fin} \vDash \varphi$.
Now, observe that $EA_{m',n'} \vDash EA_{m,n}$ if min$\{m',n'-m'\} \geq$ max$\{m,n-m\}$. Thus, if we take $k$ $=$ $\max_{EA_{m,n} \in \mathit{Fin}}{\{\max\{{m,n-m}\}\}}$, for any $EA_{m,n}\in \mathit{Fin}$, $EA_{k,2k} \vDash EA_{m,n}$ as $k\geq \max\{m,n-m\}$. Choosing $k$
 accordingly, and using the fact that $\mathit{Fin} \vDash \varphi$, we have that 
$EA_{k,2k} \vDash \varphi.$ Then, $\mu(\varphi) \geq \mu(EA_{m,n}) = 1$, that is, $\mu(\varphi) = 1.$
 Thus, for any sentence $\varphi$ such that $\mathsf{RG} \vDash \varphi$, $\varphi$ is almost surely true. 

\end{proof}


\section{Hiding on random graphs with constant edge probability}

Let us now use the proof discussed above to shed light on winning or losing in cops and robber games played on random graphs, when the number of vertices of the graph tends to infinity. As mentioned in the introduction, it would be natural to assume that \emph{the bigger the graph is, easier it is for the robber to hide in there}. In what follows, we provide a formal seal to this intuition that we have and also explore cases where the results may not agree with such instinctive way of thinking. 

We assume here that for any edge of the graph, the probability that the edge appears in the graph is $p$, where $0 < p < 1.$ We analyze what happens for the graph $G(N,p)$, where $N$ is the number of vertices and $p$ is the constant probability of an edge appearing that does not depend on $N$. We showed above that any FOL definable graph property is true in almost all graphs if it is true in the random graph $G(N,p)$ with $N \to \infty$ and $p$ constant. We considered $p = 1/2$ in the result proved in the previous section, but the same result is applied for any $p$ with $0 < p < 1.$ Let us first look at the game between a single cop and a robber played on random graphs $G(N,p)$.  

\begin{theorem}\label{one-cop}
 In a random graph $G(N,p)$ where p is the constant probability of an edge appearing in the graph and $N\rightarrow \infty$, the robber almost always has a winning strategy. 
\end{theorem}

\begin{proof}
    Let us first define a neighborhood $\mathsf{N}(u)$ of a vertex $u$ in a graph by: $\mathsf{N}(u):= \{v : E(u,v)\}$, with $E$ denoting the binary edge relation on a graph.  For a single cop and a robber, if x is the position of the cop and y is the position of the robber, then the cop can win in the next move iff $\mathsf{N}(y) \subseteq \mathsf{N}(x)$. 
Similarly, the robber can prevent a win for the cop iff $\exists z$ such that $z \in \mathsf{N}(y)$ and $z \notin \mathsf{N}(x)$. Let us consider the extension axioms now.

     $$EA_{m,n} : \forall x_1 ,x_2 ...x_n (\bigwedge \limits_{1 \leq i < j \leq n} \neg (x_i= x_j) \rightarrow
      \exists z (\bigwedge \limits_{1 \leq i \leq n} \neg (z= x_i) \wedge \bigwedge \limits_{1\leq i \leq m} E(z,x_i) \wedge \bigwedge \limits_{m < i \leq n} \neg E(z,x_i)) )$$

It can be shown as earlier that the extension axioms hold in the random graph $G(N,p)$, as $N\rightarrow \infty$. If we take any n vertices, the probability that a vertex is adjacent or not adjacent to these vertices in the desired way is given by $p^m \times (1-p)^{n-m}$. Thus, the probability of the corresponding extension axiom not being true is given by
$\lim \limits_{N \to \infty}(1- (p^m \times (1-p)^{n-m}))^N = 0.$
Now, we claim that for any positions of the cop and the robber, the robber can avoid the cop in the next move, that is, for any two vertices $x$ and $y$, there is a vertex $z$, adjacent to $y$ but not adjacent to $x.$ This statement is equivalent to the following sentence in first-order language:
\begin{equation*}
    \forall x \forall y \exists z (\neg(x= y) \rightarrow E(y,z) \wedge \neg E(x,z))
\end{equation*}
 \begin{figure}[ht]
       \centering
       \begin{tikzpicture}[every node/.style={circle, draw, minimum size=1cm}]
  
  \node[fill=DarkOrchid!50, label=above:Cop] (x) at (0,2) {x};
  \node[fill=ForestGreen!50, label=above:Robber](y) at (4,2) {y};
  \node (z) at (2,0) {z};

  \draw (y) -- (z);
  \draw[dotted] (x) -- (z);
  
\end{tikzpicture}
\caption{The \textcolor{ForestGreen}{robber} at $y$ has a move to a vertex $z$ avoiding the \textcolor{DarkOrchid}{cop} at $x.$}
\end{figure}
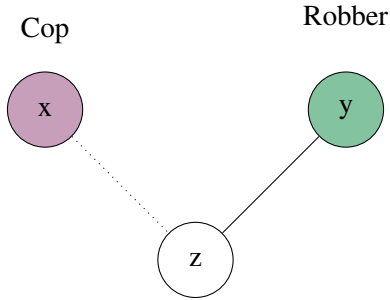

\noindent We have that 
\begin{equation*}
    \forall x \forall y \exists z (\neg(x= y) \rightarrow \neg(z=x) \wedge \neg (z=y) \wedge E(y,z) \wedge \neg E(x,z))
\end{equation*}
is an extension axiom, and it holds in the random graph.
This implies that at any point, there is always a vertex $z$, such that the robber has a move from its current position $y$ to the position $z$, in which case he can avoid the cop whose current position is $x$, as $x$ is not adjacent to $z$. Moreover, as extension axioms are almost surely true, we can say that robber has a winning strategy in almost all possible graphs.
\end{proof}

We now show that it does not help if you have a million cops on the lookout for a single robber when the game arena gets large enough, that is, even in the case for $k > 1$ cops and a single robber, the robber would still be able to avoid the cops.

\begin{theorem}\label{k-cop}
 In a random graph $G(N,p)$ where p is the constant probability of an edge appearing in the graph and $N\rightarrow \infty$, the robber almost always has a winning strategy, playing against $k > 1$ cops.
\end{theorem}

\begin{proof}
We can show that the robber has a winning strategy by considering the following extension axiom:

    $$\forall x_1 \forall x_2 ... \forall x_k \forall y \exists z ((\bigwedge \limits_{1\leq i < j \leq k} \neg(x_i=x_j)\wedge \bigwedge\limits_{i=1}^{k} \neg (x_i=y)) \rightarrow \bigwedge \limits_{i=1}^{k} \neg (x_i=z) \wedge \neg (y=z) \wedge \bigwedge \limits_{i=1}^{k} \neg E(x_i,z)\wedge E(y,z))$$

Arguing in a similar manner, we can say that the robber has a winning strategy on a random graph, as at any point there is a vertex $z$, such that the robber with its current position $y$, can move to $z$, and the cops whose current positions are $x_1, \ldots, x_k$ cannot reach $z$, as z is not adjacent to any one of them.
\end{proof}

Combining the results above, we can make the following statement regarding cops and robber game played on random graphs.

\begin{Corollary}\label{cor-1}
    The robber almost always has a winning strategy in cops and robber game played on random graphs with $k \geq 1$ cops and a single robber.
 \end{Corollary}

We have dealt with the general cops and robber games played on random graphs where edges occur with a constant probability $p.$ Now we move on to certain variants of the game played on large random game graphs to investigate the possibility of cops catching the robber in these variants. 
 
 \subsection{Cops and robber games with traps and roadblocks}

 Among the variants, we start with a lesser known one. The main idea is that of \emph{traps} \cite{traps,cop-robber-book} that can be placed on any vertex by a cop. If the robber ever occupies the same vertex, then he is caught. 
 The cops can set those traps in some vertices when they arrive and can also remove when one of them arrives again at that vertex and reuse that trap. The robber gets caught if he meets either any of the cops or reaches any of the traps. In some sense, cops may have more power when we consider $m$ cops with $n$ traps and a single robber. However, that is not always the case: if we consider the complete bipartite graph $K_{3,3}$, 2 cops can win the corresponding game, whereas 1 cop having the ability to put 1 trap cannot win the game. Evidently, even with all these traps, the cops cannot win as we see below.
 
 \begin{theorem}\label{trap}
 In a random graph $G(N,p)$ where p is the constant probability of an edge appearing in the graph and $N\rightarrow \infty$, the robber almost always has a winning strategy, playing against $m$ cops equipped to put $n$ traps.
     
 \end{theorem}
 \begin{proof}
     Let $x_1,x_2,...,x_m$ be the positions of the cops and $t_1,t_2,...,t_k$ be the positions of the traps at some points where $k \leq n$, and $y$ be the position of the robber. Consider the following extension axiom:
     
         $$\forall x_1... \forall x_m \forall t_1...\forall t_k \forall y \exists z \Bigl( \bigl( \bigwedge \limits_{\substack{1 \leq i \leq m\\ 1\leq j \leq k}} (\neg (x_i = t_j) \wedge \neg (y = x_i) \wedge \neg (y = t_j)) \wedge \bigwedge \limits_{\substack{i \neq j\\ i,j= 1}}^{m} \neg (x_i=x_j) \wedge \bigwedge \limits_{\substack{i \neq j\\ i,j= 1}}^{k} \neg (t_i=t_j) \bigr)\rightarrow $$

\begin{flushright}
    $\bigl(\bigwedge \limits_{i=1}^m \neg (z=x_i)\wedge \bigwedge \limits_{i=1}^k \neg (z=t_i) \wedge \neg(y=z) \wedge \bigwedge \limits_{i=1}^m \neg E(z,x_i) \wedge E(y,z) \bigr) \Bigr)$
\end{flushright}

If there is any such $z$, then robber can move to $z$, as they are adjacent. And he can avoid all cops as no cop position is adjacent to $z$, and he can also avoid traps as there is no trap position is equal to $z$. Since extension axioms hold in random graph, we have that the robber can avoid crops and traps at any point in that graph, and thus the robber has a winning strategy.
  \end{proof}

 We can also consider a similar variant where a finite number of \emph{road-blocks} can be placed on the edges by cops \cite{cop-robber-book}. The cops can move over those blocked edges but the robber cannot. These road-blocks can only be put on edges when a cop is at one of the endpoints of the edge. As in the case of traps, the cops get empowered once again, but still cannot win. The results follow in the same way as in the case of the games with traps, as placing a road-block has the same effect as placing a trap on one of the end points between which the road-block is put.

 \subsection{Cops and Robbers with different edge set}

In the next variant of the game that we consider, the cops and the robber get different edge sets to move -- the cops get the complement of the robber's edge set. Without loss of generality we will consider a single cop here. We consider the given graph to be the edge set of the robber, and its complementary set provides the cop's edges to move.

\begin{theorem}\label{different}
 In a random graph $G(N,p)$ where p is the constant probability of an edge appearing in the graph and $N\rightarrow \infty$, where cops and robber move along complementary edge sets, the robber almost always has a winning strategy.
 \end{theorem}

\begin{proof}
    The proof is similar to the one cop and one robber game, the only subtle difference is in the extension axiom:
$$
   \forall x \forall y \exists z (\neg(x=y) \rightarrow E(y,z) \wedge E(x,z))
$$  
   The statement says that for any cop position $x$ and any robber position $y$, there exists a move from $y$ to $z$, where x cannot reach, as the cop and the robber move in complementary edge sets, and $E$ corresponds to robber's moves in particular.
As earlier, we can say that robber has a winning strategy in almost all such games. Note the difference in the axiom from the one cop - one robber case.
\end{proof}

 \subsection{Cops and robber game with tandem-cops}
 
The final variant that we will consider is a game played by tandem-cops. In this case, two cops move together \emph{in tandem}, that is, they can either be at the same vertex or at vertices adjacent to each other, and this condition continues to hold as they move along the edges of the graph. Such a pair of cops constitute the notion of tandem-cops. Thus, when one of the tandem-cops moves to a vertex, the other one may to move to the same vertex or to one that is adjacent to the former. If we consider two cops, say $C_1$ and $C_2$, without loss of generality, we can have $C_1$ moving first to a vertex $x_1'$, say, which is adjacent to his previous position. Then the other cop $C_2$ can move to $x_1'$ or any other vertex $x'_2$, say, that is adjacent to $x_1'$.

A game is \emph{tandem-win} if one pair of tandem-cops suffices to capture the robber. We note that the class of games where the cops win is a proper subset of the class of games where the tandem-cops win. If we consider the cycle with 4 vertices, $C_4$, one can easily check that it is a tandem-win, whereas, in the case of one cop, the robber has a winning strategy by taking his position at the vertex diagonally opposite to that of the cop. Consider the  \textit{Petersen Graph} given below, where at least three cops are needed to beat the robber. Note that the graph has diameter two; that is, any two vertices are at a distance of at most two units. As tandem-cops can cover two units of distance in a single move, the game on this graph is tandem-win. In the following, we check whether tandem-cops provide any advantage over the robber in random graphs. Finally, they do!

\begin{figure}[ht]
    \centering
    \begin{tikzpicture}[scale=2]
  
  \foreach \i in {0,...,4} {
    \node[draw, circle] (A\i) at ({90 - \i*72}:1) {\i};
  }

  \foreach \i in {5,...,9} {
   
    \node[draw, circle, ] (A\i) at ({90 - \i*72}:0.5) {\i};
  }
 \draw[-] (A0) -- (A1);
  \draw[-] (A1) -- (A2);
  \draw[-] (A2) -- (A3);
  \draw[-] (A3) -- (A4);
  \draw[-] (A0) -- (A4);  

  \foreach \i in {0,...,4}
  {
  \pgfmathtruncatemacro{\j}{\i+5}
  \draw[-] (A\i) -- (A\j);
  }
  \draw[-] (A5) -- (A7);
  \draw[-] (A5) -- (A8);
  \draw[-] (A6) -- (A8);
  \draw[-] (A6) -- (A9);
  \draw[-] (A7) -- (A9);

\end{tikzpicture}
    \caption{Petersen Graph}
\end{figure}
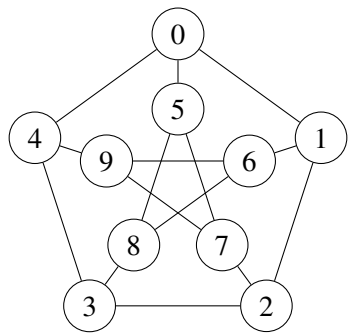
 
\begin{theorem}\label{tandem}
 In a random graph $G(N,p)$ where p is the constant probability of an edge appearing in the graph and $N\rightarrow \infty$, cops almost always have a winning strategy whenever there is a pair of tandem-cops playing.
 \end{theorem}

\begin{proof}
    Suppose $x_1$ and $x_2$ are the positions of the tandem-cops. Then there is an edge between $x_1$ and $x_2$.\\
    Let $y$ be the position of the robber. We consider the following extension axiom:
   
        $$\forall x_1 \forall x_2 \forall y \exists z ((\neg (x_1 =x_2) \wedge \bigwedge \limits_{i=1,2}\neg (x_i=y)) \rightarrow (\bigwedge \limits_{i=1,2} \neg (x_i=z)\wedge \neg(y=z) \wedge E(x_1,z) \wedge E(y,z) ))$$

    \begin{figure}[ht]
       \centering
\begin{tikzpicture}[every node/.style={circle, draw, minimum size=1cm}]
  \node[fill=DarkOrchid!50, label={[yshift=-0.3cm] Cop 1}] (x) at (0,2) {$x_1$};
  \node[fill=DarkOrchid!50, label={[yshift=-0.3cm] Cop 2}] ($x_2$) at (2,2) {$x_2$};
  \node[fill=ForestGreen!50, label={[yshift=-0.3cm] Robber}] ($y$) at (4,2) {y};
  \node (z) at (2,0) {z};
  \draw[-] (x) -- (z);
  \draw[-] (y) -- (z);
  \node (x1) at (6,2) {$x_1$};
  \node ($x_2$) at (8,2) {$x_2$};
  \node[fill=teal!50, label={[yshift=-1cm] Cop 2 / Robber}] (y1) at (10,2) {y};
  \node [fill=DarkOrchid!50, label={[yshift=-2cm]Cop 1}](z1) at (8,0) {z};
  \draw[-] (x1) -- (z1);
  \draw[-] (y1) -- (z1);
\end{tikzpicture}
\caption{The \textcolor{DarkOrchid}{tandem-cops} can catch the \textcolor{ForestGreen}{robber} in a single move.}
\end{figure}
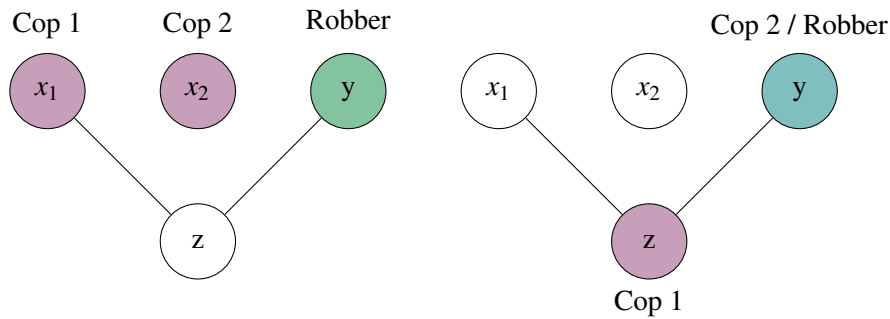 

If there is a vertex $z$  adjacent to both $x_1$ and $y$, we can move one of the tandem-cops, at $x_1$ say, to $z$ and as $y$ is adjacent to $z$, we can move the other tandem-cop to $y$ and catch the robber. And as we know, extension axioms are true in random graph $G(N,p)$ with $N \to \infty$ and $p$ constant, we can conclude that the cops have a winning strategy in the corresponding game.
    Moreover, as extension axioms are true in almost all possible graphs, cops have a winning strategy in almost all cops and robber games in presence of tandem-cops. 
\end{proof}

A natural question to ask: what is so special about tandem-cops? In whatever ways these cops move, they stay in two adjacent vertices. When one of the cops at $x_1$, say, move to a neighboring vertex, say $x'_1$, the other cop can move to any neighbor of $x'_1$, say $x'_2.$ This way it is guaranteed that the cops can cover twice the distance, as the distance between $x_1$ and $x'_2$ is 2 units. In a random graph $G(N,p)$ where $0<p<1$ is the constant edge probability, the distance between any two vertices is at most $2$ almost surely -- once again, an application of the extension axioms, which gives the cops an advantage. We should mention here that in a connected graph, a cop who can move two steps at a time always has a winning strategy against the robber, as he can continue reducing the distance between them. Using the fact that a large random graph is \emph{almost surely} connected, we can get the same result: a large random graph is \emph{almost surely} a tandem-win graph.

To finish this discussion, we note that we can similarly deal with other variants of the cops and robber game present in the literature, and show the ubiquitous applications of the extension axioms with respect to various studies on random graphs. Let us now move on to the case of random graphs with varying edge probability, which has its own nuances.


\section{Hiding on random graphs with varying edge probability}

In the previous section, we have considered random graphs $G(N, p)$, where $0<p<1$ is constant. Now, we will focus on the graphs where the edge-assigning probability depends on the graph size $N$, given by $p(N)$, say, a function of $N$. We assume here that for any edge of a graph of size $N$, the probability that the edge appears in the graph is $p(N)$, and we analyze what happens for the cops and robber game played over the random graph $G(N,p(N))$, as $N \to \infty.$ We first look at the game between a single cop and a robber played on random graphs $G(N,p(N))$, and explore the values for $p(N)$ for which the cop or the robber has a winning strategy as $N \to \infty.$ We have already seen that for the constant function $p(n)=p$ where $0<p<1$, the probability that the robber has a winning strategy tends to 1 as $N \to \infty$. 

In what follows, we consider bounds over the edge-probability function $p(N)$ explore their effects on the winning strategies of players in the game. To this end we define the threshold functions \cite{threshold} for properties in random graphs.

\begin{definition}[Threshold function]
    Let $G(N,p(N))$ be a random graph with varying edge probability and $\mathsf{P}$ be a graph property. We call a function $f(N)$ threshold function for $\mathsf{P}$ if:
    \begin{enumerate}
    \item whenever $p(N) \ll f(N)$, $\mu_G(\mathsf{P}) = 0$.
    \item whenever $f(N) \ll p(N)$, $\mu_G(\mathsf{P}) = 1$.
    \end{enumerate}
    where, $a(n)\ll b(n)$ means that $\lim\limits_{n \rightarrow \infty}\frac{a(n)}{b(n)} = 0$
\end{definition}

We give below a sufficient condition for the existence of threshold functions. Before doing that, let us first introduce the notion of edge-monotone property.

\begin{definition}[Edge-monotone property] A graph property $\mathsf{P}$ is said to be edge-monotone increasing (decreasing) if for any graph $G$ satisfying $P$, a graph $G'$ obtained from  $G$ by adding (deleting) an edge would also satisfy $\mathsf{P}.$ 
\end{definition}

\begin{theorem}[\cite{threshold}]\label{threshold}
    Any edge-monotone graph property always admits a threshold function.
\end{theorem}

Thus, if we can show that the property of a player winning the one cop and one robber game is an edge-monotone property, that would ensure the existence of threshold functions. However, that is not the case. 

\begin{proposition}\label{monotone}
 The property of a player winning the game of one cop and one robber played on a graph is neither edge-monotone increasing nor edge-monotone decreasing.   
\end{proposition} 
\begin{proof}
    Without loss of generality, consider the case of the robber. Let us call the property of the robber having a winning strategy as $Win_{\mathsf{R}}$.
    Consider the  graph $C_4$, a cycle with 4 vertices. In this graph, the robber can always evade the cop by choosing to stay at a vertex that is diagonally opposite to that of cop. Thus, $C_4$ satisfies $Win_\mathsf{R}.$ Now, if we add a diagonal edge, we will get a diamond $D_4$, where if the cop chooses to stay at any of the vertices with degree 3, she can catch the robber in the next move. Thus, $D_4$ does not satisfy $Win_\mathsf{R}.$ Similarly, if we remove any edge from $C_4$, we will get a four-vertex path $P_4$. In $P_4$, the cop has a winning strategy.
    Hence, the property of the robber winning the game is neither edge-monotone increasing nor edge-monotone decreasing. Similarly, the property of the cop winning the game is neither edge-monotone increasing nor edge-monotone decreasing.
    \end{proof}

    Let us now get back to the extension axioms to see whether those ideas can shed some light on the threshold functions. As we have shown previously in Theorem \ref{one-cop}, $EA_{1, 2}:= \forall x \forall y \exists z (\neg(x= y) \rightarrow E(y,z) \wedge \neg E(x,z))$ is almost surely true,
and thus robber has a winning strategy in almost all possible one cop and one robber games played on random graphs with constant edge probability. Now, with the varying edge probability $p(N)$, 
the probability that $EA_{1, 2}$ is satisfied is given by $1- (1-p(N)(1-p(N)))^N.$ Thus, if $\lim \limits _{N \to \infty} (1-p(N)(1-p(N)))^N = 0$, then the robber has a winning strategy in $G(N,p(N))$ almost surely and also in $G(N,(1-p(N))).$ Thus, for the property $EA_{1, 2}$, we cannot have a threshold function, and so we cannot move forward as in Section 3. 
Even though the extension axioms may not work for getting threshold functions we have the following result which serves our purpose to some extent.

\begin{theorem}[\cite{ShelahJoel1988}]\label{varying1}
Any sentence in the first-order theory of graphs is 
either almost surely true or almost surely false in random graphs $G(N,p(N))$, with $p(N)$ satisfying any one of the following conditions.
\begin{enumerate}
    \item $p(N) \ll N^{-2}$
    \item for some integer $k$, $N^{-1-\frac{1}{k}} \ll p(N) \ll N^{-1-\frac{1}{k-1}}$
    \item $N^{-1-\epsilon} \ll p(N) \ll n^{-1}$ for all $\epsilon>0$
    \item $N^{-1} \ll p(N) \ll N^{-1}\log N$
\end{enumerate}
\end{theorem}

The proof idea is quite similar to that of Theorem \ref{fol}, the only difference being that the set of axioms that is considered in each of the above cases is somewhat more involved. For a detailed proof, see \cite{ShelahJoel1988}. Let us now apply this result to find the threshold functions corresponding to the existence of winning strategies in the single cop and single robber game. We use relevant axioms used in the proof of the theorem above.

\begin{theorem}\label{one-cop-vary}
In a one cop and one robber game played on a random graph $G(N, p(N))$ where $p(N)$ is the varying edge probability and $N \rightarrow \infty$, the robber almost always has a winning strategy, whenever any one of the four conditions of Theorem \ref{varying1} holds.
\end{theorem}
\begin{proof}
We prove the result as follows: (1) For $p(N) \ll N^{-2}$, consider the formula $\forall x \forall y \neg (E(x,y))$ which is almost surely true. Evidently, the robber has a winning strategy. (2) For $N^{-1-\frac{1}{k}} \ll p(N) \ll N^{-1-\frac{1}{k-1}}$, consider the formula $\varphi_{t,r}$ saying that: there are at least $r$ components $t$, $t$ being a tree on at most $k$ points. For example, if we take $t_1$ as the tree with one vertex, then the formula
$\varphi_{t_1,2}$ is given by $\exists x_1 \exists x_2\forall y(((y=x_1) \vee \neg (E(y,x_1))) \wedge ((y=x_2) \vee \neg (E(y,x_2))))$. Now, $\varphi_{t,r}$ is almost surely true for any $r$ and any such tree $t$, in particular $\varphi_{t,2}$. This implies that the robber has a winning strategy as there is more than one component, almost surely. (3) For $N^{-1-\epsilon} \ll p(N) \ll N^{-1}$ for all $\epsilon>0$, the proof is similar to the case above, based on the corresponding axioms discussed in \cite{ShelahJoel1988}. (4) For $N^{-1} \ll p(N) \ll N^{-1}\log N$, the proof is once again similar.

\end{proof}

Along similar lines we have the following result for the case of $k > 1$ cops and 1 robber. 

    \begin{theorem}\label{k-cop-vary}
In a $k >1$ cops and one robber game played on a random graph $G(N, p(N))$ where $p(N)$ is the varying edge probability and $N \rightarrow \infty$, the robber almost always has a winning strategy, whenever each of the four conditions of Theorem \ref{varying1} holds.
\end{theorem}
\begin{proof}
The proofs for all the cases are as above with necessary modifications with respect to the number $k$. As an example, for the second case, we need to consider $k + 1$ components instead of $2$ components that we did for the one-cop case. That is, we need to consider $\varphi_{t,k+1}$, and the proof would go through as earlier. To give an example as earlier, if $t_1$ denotes the tree with one vertex,  $\varphi_{t_1,k+1}$ is given by:  $\exists x_1 \exists x_2...\exists x_k\exists x_{k+1}\forall y(\bigwedge \limits_{i=1}^{k+1}((y=x_i) \vee \neg (E(y,x_i))))$. 
\end{proof}

\subsection{Other variants}

Similar results can be proved for the \emph{cops and robber game with traps} and \emph{cops and robber games with roadblocks} played on random graphs $G(N, p(N))$ with varying edge probability and $N \rightarrow \infty,$ and we leave it to the readers. The case of games with \emph{different edge sets} played on random graphs with varying edge probabilities is more involved as we see below. 

For the complementary edge sets, we basically consider a graph and its edge-complement graph. Without loss of generality, we assume that the robber is traveling through the original graph, and the cop is traveling through the complement. Evidently, the property $Win_{\mathsf{R}}$ is monotone increasing, which implies that it would always have a threshold function.
As in the constant probability case, the robber wins whenever the underlying graph satisfies the extension axiom: 
$\forall x \forall y \exists z (\neg(x=y) \rightarrow E(y,z) \wedge E(x,z)).$ Our aim is now to compute the threshold function. To this end, let us first introduce the following concepts that are discussed in \cite{threshold-ext}.

\begin{definition}[Rooted graph]
    A graph $H = (V_H, E_H)$ can be considered as a rooted graph $(H, R)$, where $R \subseteq V_H$ is called the root.
\end{definition}

\begin{definition}[Extension statement]
    A graph $G$ is said to satisfy Extension statement $(Ext(H, R))$ where $R=\{x_1, \ldots, x_r\}$ and $V_H = \{x_1, \ldots, x_r, y_1, \ldots, y_s\}$, if for all choices of $\{x'_1, \ldots, x'_r\} \subseteq V_G$, we can find $\{y'_1, \ldots, y'_s\} \subseteq V_G$ such that whenever $(x_i,y_j) \in E_H$, $(x'_i,y'_j) \in E_G$ and whenever $(y_i,y_j) \in E_H$, $(y'_i,y'_j) \in E_G$.
\end{definition}

We note that a graph $G$ satisfying the formula 
$\forall x\forall y \exists z (E(x,z) \wedge E(z,y))$ is the same as $G$ satisfying the extension statement $Ext(H, R)$ where $R=\{x,y\}$, $V_H = \{ x,y,z\}$ and $E_H = \{(x,z),(y,z)\}$. And as discussed in Theorem \ref{different}, whenever a graph satisfies $Ext(H, R)$, the robber has a winning strategy.

Some more notions from \cite{threshold-ext} are introduced for the sake of completion: (1) {\bf Subextension:} We call a rooted graph $(S,R)$ a subextension of a rooted graph $(H,R)$ if $S$ is an induced subgraph of $H$. It is proper if $V_S$ is a proper subset of $V_H$. (2) {\bf Density of a rooted graph:} Density of a rooted graph $dens(H,R)$ is given by $\frac{|E_H \setminus E_R|}{|V_H|}$. (3) {\bf Strictly balanced:} A rooted graph $(H,R)$ is called strictly balanced if for all subextensions $(S,R)$, $dens(S,R) < dens(H,R)$. (4) {\bf Maximal average degree:} A maximal average degree of a rooted graph $(H,R)$ is given by  $mad(H,R):= max\{dens(S,R) | (S,R)$ is a subextension of $(H,R)\}$ (5) {\bf Primal subextension:} We call  a subextension $(S,R)$ of $(H,R)$ a primal subextension if $dens(S,R)= mad(H,R)$. (6) {\bf Grounded subextension:} We call a subextension grounded if there is at least one edge between a root and a non-root vertex. All these concepts are used in the following result that provides us with the required threshold function.

\begin{theorem}[\cite{threshold-ext}] For a rooted graph $(H,R)$, the following holds:
    (i) If no primal subextension of $(H,R)$ is grounded, then $f(n)=n^{-\frac{1}{mad(H,R)}}$is a theshold function for the property $Ext(H,R).$
    (ii) If there are grounded primal subextensions and $s$ is the smallest value of $|E_S \setminus E_R|$ over all such subextensions $S,$ then $f(n) = n^{-\frac{1}{mad(H,R)}} (\log n)^{\frac{1}{s}}$ is a threshold function for $Ext(H,R).$
\end{theorem}

We now have the following threshold for the cops and robber game with complementary edge sets.

\begin{theorem}
In a cops and robber game with complementary edge sets played on a random graph $G(N, p(N))$ where $p(N)$ is the varying edge probability and $N \rightarrow \infty$, the robber has a winning strategy almost surely, whenever  $N^{-\frac{3}{2}}(\log N)^{\frac{1}{2}} \ll p(N).$ Similarly, the cop has a winning strategy almost surely, whenever $p(N) \ll N^{-\frac{3}{2}}(\log N)^{\frac{1}{2}}.$ 
\end{theorem}
\begin{proof}
We prove below the result for the robber. In this case, we have the extension axiom \small{$\forall x\forall y \exists z (E(x,z) \wedge E(z,y))$}, and corresponding to the axiom, we have the extension statement $Ext(H, R)$ where $R=\{x,y\}$, $V_H = \{ x,y,z\}$ and $E_H = \{(x,z),(y,z)\}$. Note that in the extension statements, for all choices of root vertices, we get a choice of $z$ such that it is incident to the required edges. Now, $(H,R)$ has only one subextension, which is $(H,R)$ itself. It is primal as well as grounded as there is an edge between a root vertex and a non-root vertex, (e.g., $(x,z)$ and $(y,z)$). Now, $mad(H,R)= \frac{|E_H\setminus E_R|}{|V_H|}=\frac{2}{3}$ and $s=|E_H \setminus E_R|=2$. As there is a primal subextension which is grounded, using the previous theorem, we get the threshold function of $Ext(H,R)= N^{-\frac{3}{2}}(\log N)^{\frac{1}{2}}$. And, as a graph satisfying $Ext(H,R)$ implies that the robber has a winning strategy, $N^{-\frac{3}{2}}(\log N)^{\frac{1}{2}} \ll p(N)$ implies that the robber has a winning strategy in $ G(N,p(N))$ almost surely. The case for the cop can be shown similarly, with respect to the extension statement given by $\neg \varphi'$, where $\varphi'$ is given by $\exists x \forall y \forall z ( \neg E(y,z) \vee \neg E(x,z)).$ 
\end{proof}

In the proof above, we have extension statements for the extension axioms but not the earlier ones, e.g., the cases of $k\geq1$ cops and robber, as these statements do not preserve the property of non-adjacency of two vertices (note that, $\neg \varphi'$ above is also all about adjacency). 

Finally, we look into the game of \emph{tandem-cops} and robber. The tandem-cops have a winning strategy if the graph satisfies the extension axiom: $\forall x_1 \forall x_2 \forall y \exists z ((\neg (x_1 =x_2) \wedge \bigwedge \limits_{i=1,2}\neg (x_i=y)) \rightarrow (\bigwedge \limits_{i=1,2} \neg (x_i=z)\wedge \neg(y=z) \wedge E(x_1,z) \wedge E(y,z) ))$. This does not contain any non-adjacency condition. Thus, following the same extension statement idea, we have:

\begin{theorem}
In a tandem-cops and robber game played on a random graph $G(N, p(N))$ where $p(N)$ is the varying edge probability and $N \rightarrow \infty$, the cops have a winning strategy almost surely, whenever  $N^{-\frac{3}{2}}(\log N)^{\frac{1}{2}} \ll p(N)$.
\end{theorem}


\section{Concluding remarks}

In conclusion, we can say that in accordance with our intuition, in almost all cases, it is possible for the robber to hide in large random graphs. The cops have a better chance in capturing the robber if they move in a synchronous way so that the distance between them always remains less than two units. In fact, the tandem-cops, who can move two units at a time, move twice the speed of the robber. If the winning condition is edge-monotone, we get a proper threshold function. The question remains for the non-monotone cases, which we would like to work on in the future. That all these results can be shown as an application of first order theory of random graphs gives us the added benefit of exploring the inherent connection between logic, probability, games and combinatorics.


\section*{Acknowledgements}

 The authors thank Johan van Benthem for a discussion on the relevant results of sabotage games, which initiated the work on this paper. They also thank the TARK 2025 reviewers for their detailed comments which helped to improve the paper.

\bibliographystyle{eptcs}
\bibliography{references}

\begin{thebibliography}{10}
\providecommand{\bibitemdeclare}[2]{}
\providecommand{\surnamestart}{}
\providecommand{\surnameend}{}
\providecommand{\urlprefix}{Available at }
\providecommand{\url}[1]{\texttt{#1}}
\providecommand{\href}[2]{\texttt{#2}}
\providecommand{\urlalt}[2]{\href{#1}{#2}}
\providecommand{\doi}[1]{doi:\urlalt{https://doi.org/#1}{#1}}
\providecommand{\eprint}[1]{arXiv:\urlalt{https://arxiv.org/abs/#1}{#1}}
\providecommand{\bibinfo}[2]{#2}

\bibitemdeclare{article}{AIGNER1984}
\bibitem{AIGNER1984}
\bibinfo{author}{M.~\surnamestart Aigner\surnameend} \&
  \bibinfo{author}{M.~\surnamestart Fromme\surnameend} (\bibinfo{year}{1984}):
  \emph{\bibinfo{title}{A game of cops and robbers}}.
\newblock {\slshape \bibinfo{journal}{Discrete Applied Mathematics}}
  \bibinfo{volume}{8}(\bibinfo{number}{1}), pp. \bibinfo{pages}{1--12},
  \doi{10.1016/0166-218X(84)90073-8}.

\bibitemdeclare{inproceedings}{vanbenthem2005}
\bibitem{vanbenthem2005}
\bibinfo{author}{J.~\surnamestart van Benthem\surnameend}
  (\bibinfo{year}{2005}): \emph{\bibinfo{title}{An Essay on Sabotage and
  Obstruction}}.
\newblock In \bibinfo{editor}{D.~\surnamestart Hutter\surnameend} \&
  \bibinfo{editor}{W.~\surnamestart Stephan\surnameend}, editors: {\slshape
  \bibinfo{booktitle}{Mechanizing Mathematical Reasoning: Essays in Honor of
  J{\"o}rg H. Siekmann on the Occasion of His 60th Birthday}},
  \bibinfo{publisher}{Springer}, pp. \bibinfo{pages}{268--276},
  \doi{10.1007/978-3-540-32254-2_16}.

\bibitemdeclare{inproceedings}{graphgame}
\bibitem{graphgame}
\bibinfo{author}{J.~\surnamestart van Benthem\surnameend} \&
  \bibinfo{author}{F.~\surnamestart Liu\surnameend} (\bibinfo{year}{2020}):
  \emph{\bibinfo{title}{Graph Games and Logic Design}}.
\newblock In \bibinfo{editor}{F.~\surnamestart Liu\surnameend},
  \bibinfo{editor}{H.~\surnamestart Ono\surnameend} \&
  \bibinfo{editor}{J.~\surnamestart Yu\surnameend}, editors: {\slshape
  \bibinfo{booktitle}{Knowledge, Proof and Dynamics}}, \bibinfo{series}{Logic
  in {{Asia}}: {{Studia Logica Library}}}, \bibinfo{publisher}{Springer}, pp.
  \bibinfo{pages}{125--146}, \doi{10.1007/978-981-15-2221-5_7}.

\bibitemdeclare{article}{bollobas}
\bibitem{bollobas}
\bibinfo{author}{B.~\surnamestart Bollob\'{a}s\surnameend},
  \bibinfo{author}{G.~\surnamestart Kun\surnameend} \&
  \bibinfo{author}{I.~\surnamestart Leader\surnameend} (\bibinfo{year}{2013}):
  \emph{\bibinfo{title}{Cops and robbers in a random graph}}.
\newblock {\slshape \bibinfo{journal}{Journal of Combinatorial Theory, Series
  B}} \bibinfo{volume}{103}(\bibinfo{number}{2}), pp.
  \bibinfo{pages}{226--236}, \doi{10.1016/j.jctb.2012.10.002}.

\bibitemdeclare{article}{threshold}
\bibitem{threshold}
\bibinfo{author}{B.~\surnamestart Bollob{\'a}s\surnameend} \&
  \bibinfo{author}{A.G. \surnamestart Thomason\surnameend}
  (\bibinfo{year}{1987}): \emph{\bibinfo{title}{Threshold functions}}.
\newblock {\slshape \bibinfo{journal}{Combinatorica}} \bibinfo{volume}{7}, pp.
  \bibinfo{pages}{35--38}, \doi{10.1007/BF02579198}.

\bibitemdeclare{book}{cop-robber-book}
\bibitem{cop-robber-book}
\bibinfo{author}{A.~\surnamestart Bonato\surnameend} \& \bibinfo{author}{R.~J.
  \surnamestart Nowakowski\surnameend} (\bibinfo{year}{2011}):
  \emph{\bibinfo{title}{The Game of Cops and Robbers on Graphs}}.
\newblock {\slshape \bibinfo{series}{The Student Mathematical
  Library}}~\bibinfo{volume}{61}, \bibinfo{publisher}{AMS},
  \bibinfo{address}{Providence}, \doi{10.1090/stml/061}.

\bibitemdeclare{article}{p-e1}
\bibitem{p-e1}
\bibinfo{author}{T.H. \surnamestart Chung\surnameend}, \bibinfo{author}{G.A.
  \surnamestart Hollinger\surnameend} \& \bibinfo{author}{V.~\surnamestart
  Isler\surnameend} (\bibinfo{year}{2011}): \emph{\bibinfo{title}{Search and
  pursuit-evasion in mobile robotics}}.
\newblock {\slshape \bibinfo{journal}{Autonomous Robots}}
  \bibinfo{volume}{31}(\bibinfo{number}{4}), pp. \bibinfo{pages}{299--316},
  \doi{10.1007/s10514-011-9241-4}.

\bibitemdeclare{article}{traps}
\bibitem{traps}
\bibinfo{author}{N.~\surnamestart Clarke\surnameend} \&
  \bibinfo{author}{R.~\surnamestart Nowakowski\surnameend}
  (\bibinfo{year}{2001}): \emph{\bibinfo{title}{Cops, robber and traps}}.
\newblock {\slshape \bibinfo{journal}{Utilitas Mathematica}}
  \bibinfo{volume}{60}, pp. \bibinfo{pages}{91--98}.

\bibitemdeclare{article}{erdos1959random}
\bibitem{erdos1959random}
\bibinfo{author}{Paul \surnamestart Erd{\H{o}}s\surnameend} \&
  \bibinfo{author}{Alfr{\'e}d \surnamestart R{\'e}nyi\surnameend}
  (\bibinfo{year}{1959}): \emph{\bibinfo{title}{On Random Graphs}}.
\newblock {\slshape \bibinfo{journal}{Publicationes Mathematicae Debrecen}}
  \bibinfo{volume}{6}, pp. \bibinfo{pages}{290--297},
  \doi{10.5486/PMD.1959.6.3-4.12}.

\bibitemdeclare{article}{Rfagin1976}
\bibitem{Rfagin1976}
\bibinfo{author}{Ronald \surnamestart Fagin\surnameend} (\bibinfo{year}{1976}):
  \emph{\bibinfo{title}{Probabilities on Finite Models}}.
\newblock {\slshape \bibinfo{journal}{The Journal of Symbolic Logic}}
  \bibinfo{volume}{41}, pp. \bibinfo{pages}{50--58}, \doi{10.2307/2272945}.

\bibitemdeclare{article}{gilbert}
\bibitem{gilbert}
\bibinfo{author}{E.N. \surnamestart Gilbert\surnameend} (\bibinfo{year}{1959}):
  \emph{\bibinfo{title}{Random Graphs}}.
\newblock {\slshape \bibinfo{journal}{The Annals of Mathematical Statistics}}
  \bibinfo{volume}{30}(\bibinfo{number}{4}), pp. \bibinfo{pages}{1141--1144},
  \doi{10.1214/aoms/1177706098}.

\bibitemdeclare{article}{p-e2}
\bibitem{p-e2}
\bibinfo{author}{M.W. \surnamestart Hasan\surnameend} \& \bibinfo{author}{L.G.
  \surnamestart Ibrahim\surnameend} (\bibinfo{year}{2024}):
  \emph{\bibinfo{title}{A pursuit-evasion game robot controller design based on
  a neural network with an improved optimization algorithm}}.
\newblock {\slshape \bibinfo{journal}{Results in Control and Optimization}}
  \bibinfo{volume}{17}, p. \bibinfo{pages}{100503},
  \doi{10.1016/j.rico.2024.100503}.

\bibitemdeclare{article}{LHS-journal}
\bibitem{LHS-journal}
\bibinfo{author}{D.~\surnamestart Li\surnameend},
  \bibinfo{author}{S.~\surnamestart Ghosh\surnameend},
  \bibinfo{author}{F.~\surnamestart Liu\surnameend} \&
  \bibinfo{author}{Y.~\surnamestart Tu\surnameend} (\bibinfo{year}{2023}):
  \emph{\bibinfo{title}{A simple logic of the hide and seek game}}.
\newblock {\slshape \bibinfo{journal}{Studia Logica}} \bibinfo{volume}{111},
  pp. \bibinfo{pages}{821--853}, \doi{10.1007/s11225-023-10039-4}.

\bibitemdeclare{book}{fmt}
\bibitem{fmt}
\bibinfo{author}{Leonid \surnamestart Libkin\surnameend}
  (\bibinfo{year}{2010}): \emph{\bibinfo{title}{Elements of Finite Model
  Theory}}.
\newblock \bibinfo{publisher}{Springer}, \doi{10.1007/978-3-662-07003-1}.

\bibitemdeclare{inproceedings}{sabotage-random}
\bibitem{sabotage-random}
\bibinfo{author}{K.~\surnamestart Mierzewski\surnameend}
  (\bibinfo{year}{2025}): \emph{\bibinfo{title}{When random graphs are safe for
  travel:a note on the sabotage game}}.
\newblock In \bibinfo{editor}{J.~\surnamestart van Benthem\surnameend} \&
  \bibinfo{editor}{F.~\surnamestart Liu\surnameend}, editors: {\slshape
  \bibinfo{booktitle}{Graph Games and Logic Design: Recent Developments and
  Future Directions}}, \bibinfo{publisher}{Springer},
  \doi{10.1007/978-981-15-2221-5_7}.

\bibitemdeclare{article}{NOWAKOWSKI1983}
\bibitem{NOWAKOWSKI1983}
\bibinfo{author}{R.~\surnamestart Nowakowski\surnameend} \&
  \bibinfo{author}{P.~\surnamestart Winkler\surnameend} (\bibinfo{year}{1983}):
  \emph{\bibinfo{title}{Vertex-to-vertex pursuit in a graph}}.
\newblock {\slshape \bibinfo{journal}{Discrete Mathematics}}
  \bibinfo{volume}{43}(\bibinfo{number}{2}), pp. \bibinfo{pages}{235--239},
  \doi{10.1016/0012-365X(83)90160-7}.

\bibitemdeclare{phdthesis}{quilliot1978}
\bibitem{quilliot1978}
\bibinfo{author}{A.~\surnamestart Quilliot\surnameend} (\bibinfo{year}{1978}):
  \emph{\bibinfo{title}{Jeux et pointes fixes sur les graphes}}.
\newblock Ph.D. thesis, \bibinfo{school}{Universit{\'e} de Paris VI}.

\bibitemdeclare{article}{ShelahJoel1988}
\bibitem{ShelahJoel1988}
\bibinfo{author}{S.~\surnamestart Shelah\surnameend} \& \bibinfo{author}{J.H.
  \surnamestart Spencer\surnameend} (\bibinfo{year}{1988}):
  \emph{\bibinfo{title}{Zero-one laws for sparse random graphs}}.
\newblock {\slshape \bibinfo{journal}{Journal of the American Mathematical
  Society}} \bibinfo{volume}{1}, pp. \bibinfo{pages}{97--115},
  \doi{10.1090/S0894-0347-1988-0924703-8}.

\bibitemdeclare{article}{threshold-ext}
\bibitem{threshold-ext}
\bibinfo{author}{J.~\surnamestart Spencer\surnameend} (\bibinfo{year}{1990}):
  \emph{\bibinfo{title}{Threshold functions for extension statements}}.
\newblock {\slshape \bibinfo{journal}{J. Comb. Theory Ser. A}}
  \bibinfo{volume}{53}(\bibinfo{number}{2}), p. \bibinfo{pages}{286–305},
  \doi{10.1016/0097-3165(90)90061-Z}.

\bibitemdeclare{article}{stojakovi2006}
\bibitem{stojakovi2006}
\bibinfo{author}{M.~\surnamestart Stojakovi{\'c}\surnameend} \&
  \bibinfo{author}{T.~\surnamestart Szab{\'o}\surnameend}
  (\bibinfo{year}{2006}): \emph{\bibinfo{title}{Positional games on random
  graphs}}.
\newblock {\slshape \bibinfo{journal}{Random Structures \& Algorithms}}
  \bibinfo{volume}{26}, \doi{10.1002/rsa.20059}.

\end{thebibliography}

\end{document}